\documentclass[letterpaper, 10 pt, conference]{ieeeconf}
\IEEEoverridecommandlockouts

\usepackage{graphics} 
\usepackage{times} 
\usepackage{amsmath} 

\usepackage{amssymb}
\usepackage{amsthm}
\usepackage{xcolor}
\usepackage{upgreek}
\usepackage{enumerate} 
\usepackage[labelformat=simple]{subcaption}
\usepackage{mathtools}
\usepackage{multirow}  
\usepackage[hidelinks]{hyperref}
\usepackage{color}
\usepackage{lettrine}
\usepackage{cite} 
\usepackage{cuted}
\usepackage{setspace}
\usepackage[poorman]{cleveref}

\newtheorem{thm}{Theorem}
\newtheorem{prop}{Proposition}
\newtheorem{definition}{Definition}

\newtheorem{lem}{Lemma}
\newtheorem{assumption}{Assumption}
\newtheorem{remark}{Remark}

\allowdisplaybreaks

\def\BibTeX{{\rm B\kern-.05em{\sc i\kern-.025em b}\kern-.08em
		T\kern-.1667em\lower.7ex\hbox{E}\kern-.125emX}}

\begin{document}

\title{Passivity-based Attack Identification and Mitigation with Event-triggered Observer Feedback and Switching Controller}

\author{Pushkal Purohit and Anoop Jain \\
\thanks{The authors are with the Department of Electrical Engineering, Indian Institute of Technology Jodhpur, Rajasthan, India 342030 (e-mail: purohit.1@iitj.ac.in, anoopj@iitj.ac.in).}
}
\maketitle


\begin{abstract}
This paper addresses the problem of output consensus in linear passive multi-agent systems under a False Data Injection (FDI) attack, considering the unavailability of complete state information. Our formulation relies on an event-based cryptographic authentication scheme for sensor integrity and considers FDI attacks at the actuator end, inspired by their practical nature and usages. For secure consensus, we propose (i) a passivity-based approach for detecting FDI attacks on the system and (ii) a Zeno-free event-triggered observer-based switching controller, which switches between the \textit{normal} and the \textit{defense} modes following an attack detection. We show that the closed-loop system achieves practical consensus under the controller's action in the defense mode. Simulation examples are provided to support the theoretical findings.
\end{abstract}


\section{Introduction} \label{Intro}
Cyber-Physical Systems (CPS) are often exposed to adversarial attacks due to their operation over insecure communication networks. In particular, FDI attacks are a major threat to the secure operation of a CPS, as they can corrupt both the actuation signal and the sensor measurements \cite{huo2022observer,joo2020resilient}. Since the sensor feedback cannot be trusted in this situation, cryptographic tools like message authentication code \cite{forouzan2008cryptography}, etc., have been popular as these guarantee sensor measurement integrity from such attacks by preventing data alteration in transit \cite{forouzan2008cryptography}. This situation is shown in Fig.~\ref{fig:attack_model}, where a cryptographic authenticator is employed in conjunction with an observer to design the secure controllers. However, continuous use of cryptographic tools can generate computation and communication overhead primarily due to the limited resources of CPS \cite{khazraei2022attack}. As a remedy, one could employ an event-based cryptographic authentication scheme where the communication is realized using sensors (such as ultrasonic and infrared) that are practically discrete or event-based \cite{carullo2001ultrasonic}. Consequently, the controller can ignore an attack injected in between adjacent events, limiting adverse effects on the system. On the other hand, actuators usually work in a continuous fashion and have faster actuation than sensor measurements (for example, the motor in Khepera IV mobile robot \cite{khepera}). Thereby, an attack injected in between adjacent events at the actuator side cannot be rejected by the actuator. For these practical reasons, it is viable to consider an event-based cryptographic authentication at the sensor end and the FDI attacks at the actuator end.

\begin{figure}[t]
	\centering
	\includegraphics[width = 0.35\textwidth]{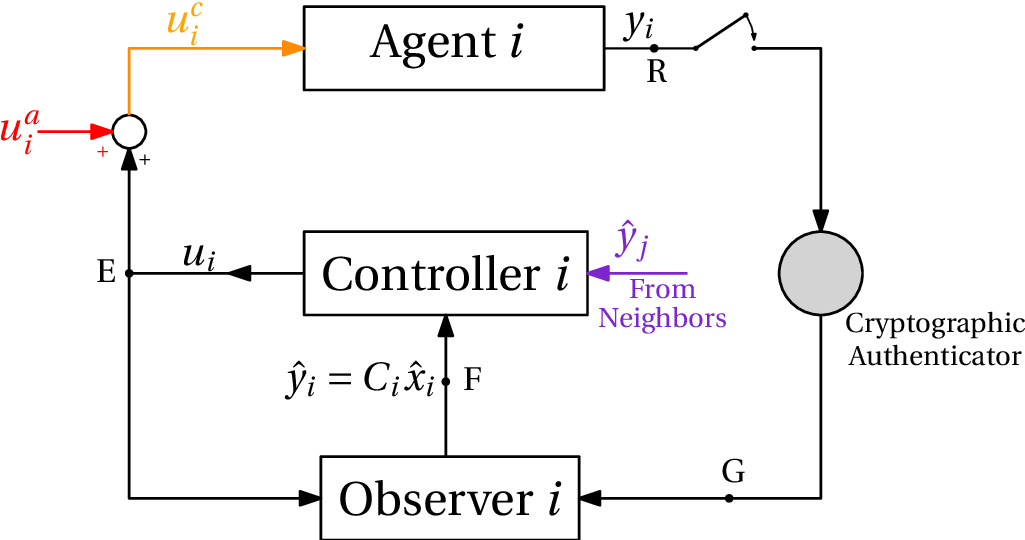}
	\caption{The $i^{\text{th}}$ agent under the FDI attack at the actuator end and with cryptographic authenticator at the sensor end.}
	\label{fig:attack_model}
\end{figure}

For the secure control of CPS, attack detection is often the foremost objective \cite{tan2020brief}. While most of the existing works focus on residual-based and estimation-based attack detectors, these might be vulnerable to destabilizing stealthy attacks \cite{khazraei2022attack,guo2022stealthy}. The authors in \cite{eyisi2014energy} presented an energy equivalence-based attack detection approach that checks for equality between the supplied, and the dissipated and stored energies. Consequently, the attacks are classified as passive and non-passive. A generalized notion of this is referred to as passivity theory \cite{khalil2002nonlinear}, which relies on an inequality between the supplied energy and the change in stored energy. Besides being a flexible tool for control design and analyzing the system's stability, the passivity property can be leveraged in identifying an attack in a CPS. Unlike \cite{eyisi2014energy}, our attack identification approach relies on a passivity \emph{inequality} between the energy supplied and the change in the system's stored energy, along with an event-based observer setting. With an observer, while an event-triggering mechanism can be implemented either by continuously monitoring the system output \cite{petri2021event} or the observer output \cite{shi2014event}, we consider the former case in this paper where the event condition is independent of the observer state in the presence of a malicious attack.

Attack mitigation is the second necessary step for securely controlling a CPS towards its desired goal following an attack detection. For this, it is essential to have (accurate) knowledge of the attack signal to mitigate its effect, usually accomplished using attack estimators designed as part of the state observers. Some popular works in this direction are \cite{meng2020adaptive,jin2017adaptive}, where the attack estimators operate continuously irrespective of the absence of an attack signal and hence, are computationally expensive \cite{yan2017resilient}. Unlike these works, this paper proposes a switching-based controller scheme employing the attack estimator only in the case of attack detection and works according to an event-triggered authentication scheme.

\emph{Main Contributions:} As shown in Fig.~\ref{fig:attack_model}, we consider a multi-agent system comprising linear passive agents interacting towards achieving output consensus over a network susceptible to malicious FDI attacks. We envisage the presence of actuator attacks $u_i^a$ while relying on an event-based cryptographic authentication for sensor measurement integrity. The proposed controller receives information about estimated states $\hat{x}_i$ from its cascaded observer and the estimated outputs $\hat{y}_j$ from the neighboring agents and works in the dual modes, namely, the \textit{normal} mode (no-attack detection) and the \textit{defense} mode (attack detection). We first obtain the Zeno-free event condition for cryptographic authentication for the stable operation of the observer, independent of the overall system's stability. We then propose a passivity-based approach for the detection of actuator attacks. It is shown that the difference between the agents' output remains bounded under the action of the proposed switching controller in the defense mode and the previously derived event condition. Alternatively, it can be said that the system achieves \emph{practical} output consensus. 

\emph{Paper Organization:} Section~\ref{Prob and model} describes the system and attack models, followed by the problem statement. Section~\ref{mitigation} introduces the Zeno-free event condition for observer stability. Section~\ref{attack_sec} discusses attack detection and mitigation by introducing the idea of the switching controller, followed by proof of the closed-loop system's stability. Section~\ref{Sim} provides simulation results before we conclude the paper and discuss the future direction of work in Section~\ref{Conclusion}.

\subsection{Notations} \label{notation}
The set of real, non-negative real numbers and positive integers is denoted by $\mathbb{R}$, $\mathbb{R}_+$, and $\mathbb{Z}_+$, respectively. $\pmb{0}_n \in \mathbb{R}^n$, $\pmb{1}_n \in \mathbb{R}^n$ are column vectors with all entries $0$, $1$, respectively. $I_n \in \mathbb{R}^{n \times n}$ is an identity matrix. The induced $2$-norm (resp., Euclidean norm) for any matrix $M$ (resp., vector $z$) is represented by $\|M\|: \mathbb{R}^{m \times n} \to \mathbb{R}_+$ (resp., $\|z\|: \mathbb{R}^{n} \to \mathbb{R}_+$). The symbol $\otimes$ denotes the Kronecker product of two matrices and $\text{diag}\{k_1, \ldots, k_n\} \in \mathbb{R}^{n \times n}$ denotes a diagonal matrix with diagonal entries $k_i$. The maximum and minimum eigenvalues of a symmetric matrix $M \in \mathbb{R}^{n \times n}$ are represented by $\lambda_{\text{max}}(M)$ and $\lambda_{\text{min}}(M)$, respectively. The Moore–Penrose pseudo-inverse of a matrix $M \in \mathbb{R}^{m \times n}$ is $M^{\dagger} \in \mathbb{R}^{n \times m}$, having the property $M M^{\dagger} M = M$. If $M$ has full column rank, $M^\dagger = (M^T M)^{-1}M^T$ and $M^\dagger M = I_n$ \cite{rao1971generalized}. For a digraph with edge set $\mathbb{E}$, $(i, j) \in \mathbb{E}$ is a directed edge from node $i$ to node $j$. The Laplacian for a graph with $N$ nodes is denoted by $\mathcal{L}_s \in \mathbb{R}^{N \times N}$ (please refer to \cite{godsil2001algebraic} for details). 

\subsection{Preliminaries}
\begin{definition}[Passive System \cite{khalil2002nonlinear}] \label{passive_def}
	Consider the system 
	\begin{equation}\label{system}
		\dot{x} =f(x, u); \quad y =h(x, u),
	\end{equation}
	with state $x \in \mathbb{R}^m$, control input $u \in \mathbb{R}^p$, and output $y\in \mathbb{R}^p$. The function $f$ is locally Lipschitz, $h$ is continuous, $f(\pmb{0}_{m}, \pmb{0}_{p}) = \pmb{0}_{m}$ and $h(\pmb{0}_{m}, \pmb{0}_{p}) = \pmb{0}_{p}$. The system \eqref{system} is said to be passive if there exists a continuously differentiable storage function $S(x): \mathbb{R}^m \to \mathbb{R}$, $S(x)\geq 0, S(\pmb{0}_m) = 0$, such that 
	\begin{equation}\label{passivity_inequality}
		u^{T} y \geq \dot{S}= ({\partial S}/{\partial x}) f(x, u), \quad \forall(x, u) \in \mathbb{R}^{m} \times \mathbb{R}^{p}.
	\end{equation}
\end{definition}
Note that the storage function $S(x)$ in \eqref{passivity_inequality} (which is also referred to as passivity inequality) is not unique. However, the quadratic storage function of the form $S(x) = (1/2) x^T x$ has its own vantages over others in analyzing linear systems, as it is computationally moderate, easy to implement, and close to the concept of energy in linear systems. Particularly, it is shown in \cite{trentelman1997storage} that any storage function can be represented as a quadratic storage function for a linear dynamical system. Motivated by this fact, we exploit the quadratic storage function in our analysis in this paper.


\begin{lem}[Implications of Passivity \cite{chopra2012output}] \label{passive}
	Consider the linear system
		\begin{equation} \label{sys}
				\dot{x} = Ax + Bu, \quad y = Cx,
		\end{equation}
	where $x \in \mathbb{R}^m, u \in \mathbb{R}^p, y \in \mathbb{R}^p$ are state, input and output vectors, and $A \in \mathbb{R}^{m \times m}, B \in \mathbb{R}^{m \times p}$ and $C \in \mathbb{R}^{p \times m}$ are system, input and output matrices, respectively. The system \eqref{sys} is said to be passive if there exists a differentiable scalar storage function $V: \mathbb{R}^m \to \mathbb{R}$ for \eqref{sys} such that $V(x) \geq 0, V(\pmb{0}_m)=0$, and $W(x) \geq 0$, such that $(\frac{\partial V}{\partial x})^T Ax = -W(x), \ (\frac{\partial V}{\partial x})^T B = y^T$.
\end{lem}

\section{System Modeling and Problem Description} \label{Prob and model}

\subsection{The Agent and Attack Models} \label{system modeling}

\subsubsection{Agent Model}
We consider a multi-agent system comprising $N$ linear passive agents under the influence of FDI attack and having combined system dynamics (see Fig.~\ref{fig:attack_model}):
\begin{equation} \label{model_complete}
\dot{x} = Ax + Bu^c, \quad y = Cx,
\end{equation}
where $x = [x_1^T, \ldots, x_N^T]^T \in \mathbb{R}^{Nm}$, $u^c = [(u^c_1)^T, \ldots, (u^c_N)^T]^T \in \mathbb{R}^{Np}$ and $y = [y_1^T, \ldots, y_N^T]^T \in \mathbb{R}^{Np}$ are the stacked state, input and output vectors, respectively, where $x_i \in \mathbb{R}^m$, $u_i \in \mathbb{R}^p$ and $y_i \in \mathbb{R}^p$ $\forall i$. Further, ${A} = \text{diag}\{A_1,\ldots,A_N\} \in \mathbb{R}^{Nm \times Nm}$, ${B} = \text{diag}\{B_1,\ldots,B_N\} \in \mathbb{R}^{Nm \times Np}$, and ${C} = \text{diag}\{C_1,\ldots,C_N\} \in \mathbb{R}^{Np \times Nm}$ are the stacked system, input and output block diagonal matrices, respectively. We incorporate the following reasonable assumptions on system \eqref{model_complete} (unless otherwise explicitly stated):
\begin{enumerate} \label{Assumption}
\item[\textbf{(A1)}] The matrix pair $(A_i, B_i)$ is controllable and $(A_i,C_i)$ is observable for all $i$.
\item[\textbf{(A2)}] The matrix $B_i$ has full column rank for all $i$.
\item[\textbf{(A3)}] The agents are strongly connected over a fixed directed network having Laplacian $\mathcal{L}_s$.
\end{enumerate}

\begin{remark}
Note that the assumption \textbf{(A2)} follows from the fact that the agents are usually under-actuated where $p < m$. If the rank of $B_i$ is less than the number of columns for an under-actuated system, there exists a set of inputs that do not affect the system dynamics, which is not practically desired. Further, in over-actuated systems, if the columns of $B_i$ are linearly dependent, some inputs will be redundant. Thus, it is reasonable to remove dependent columns and corresponding inputs with the assumption that $B_i$ has full column rank (please refer to \cite[Section~6.2.1]{chen1999linear}). Some of the most widely used robotic models, like quadcopter and car-like models, among others, satisfy \textbf{(A2)} \cite{yan2017resilient,lynch2017modern}. Other examples are$-$ tugboats \cite{ihle2007passivity}, which follow second-order Lagrangian dynamics, and harmonic oscillators \cite{xia2016synchronization}. Both these systems are passive and satisfy the preconditions \textbf{(A1)} and \textbf{(A2)}.
\end{remark}

\subsubsection{FDI Attack Model}
The FDI attack on the $i^\text{th}$ agent can be modeled as shown in Fig.~\ref{fig:attack_model}, implying that the compromised signal reaching its actuator can be written as
\begin{equation}\label{attack_ip}
u^c_i = u_i + u^a_i,
\end{equation}
for each $i$, where $u_i$, $u^a_i$, and $u^c_i$ are designed control, actuator attack, and compromised control signals, respectively. As discussed in Section~\ref{Intro}, since we employ an event-based cryptographic authentication scheme for the sensor measurements collected at point R in Fig.~\ref{fig:attack_model}, the attack at the sensor end can be neglected for a suitably chosen event condition such that the observer error dynamics remains stable. The following reasonable limitations are considered on the actuator attack signal:
\begin{assumption} \label{attack_assum}
The attack signal $u^a {=} [(u_1^a)^T, \ldots, (u_N^a)^T]^T$ $\in \mathbb{R}^{Np}$ is bounded and has a bounded derivative, that is, $\|u^a\| \leq \bar{u}^a$ and $\|\dot{u}^a\| \leq \tilde{u}^a$, for some $\bar{u}^a, \tilde{u}^a \in \mathbb{R}_+$.
\end{assumption}

Assumption~\ref{attack_assum} is motivated by the fact that the attack signal cannot be arbitrarily large, as the attacker has limited resources \cite{peng2020switching}. Secondly, arbitrarily large attack signals are easy to detect, while the attacker generally wants to remain stealthy. Further, this paper does not deal with sensor noise. In practice, noise and attacks have inherently different characteristics; noise is caused randomly and unintentionally, while attacks are injected \emph{intentionally} by an intruder to mislead or even paralyze the whole network’s behaviors while remaining undetected \cite{joo2020resilient}.

\subsection{Event-based Observer and Switching Controller}

\subsubsection{The Observer}
Due to event-triggered sensor feedback, the latest measurements $y$ are not always available at point G (in Fig.~\ref{fig:attack_model}). Let the sensor measurement at the last event $t_k$ be $\bar{y} = y(t_k), \ k \in \mathbb{Z}_+$. We define an error variable $e(t)=\bar{y} - y(t), \ t \in [t_k, t_{k+1}), \forall k$. Clearly, $e(t_k) = \pmb{0}_{Np}, \forall k$. With the error term $e$, the event-triggered observer dynamics is given by:
\begin{equation} \label{observer}
\dot{\hat{x}}  = A \hat{x} + B (u + \hat{u}^a) + H(y-\hat{y} + e), \quad \hat{y} = C\hat{x},
\end{equation}
where $H \in \mathbb{R}^{Nm \times Np}$ is the observer gain matrix, $\hat{x}$ is the estimated state and $\hat{u}^a$ is the estimated input attack signal proposed as:
\begin{equation} \label{attack_estimate}
\hat{u}^a = B^\dagger (\dot{\hat{x}} - A\hat{x} - Bu + B\hat{u}^a(t - t_d)),
\end{equation}
where $t_d > 0$ is a small constant reaction time, representing the time required to observe the effect of control $u(t)$ on the sensor output. Clearly, $\hat{u}^a$ is Lipschitz under Assumption~\ref{attack_assum}, i.e.,
\begin{equation} \label{lipsch}
	\|\hat{u}^a (t) - \hat{u}^a (t-t_d)\| \leq \psi t_d,
\end{equation}
for some constant $\psi \in \mathbb{R}_+$. Ideally, $t_d \to 0$, however, it may not be feasible in practice due to the components' response time and can be assumed sufficiently small. Next, we define the state estimation error $\xi(t) = x(t)- \hat{x}(t)$, whose time-derivative, using \eqref{model_complete} and \eqref{observer}, is obtained as
\begin{equation} \label{state_error}
\dot{\xi} = \dot{x} - \dot{\hat{x}} = (A-HC) \xi + B(u^a - \hat{u}^a) - He.
\end{equation}
In the absence of an attack (i.e., $u^a = \pmb{0}_{Np}$), attack estimator is inactive, that is, $\hat{u}^a \coloneqq \pmb{0}_{Np}$, and hence, \eqref{state_error} can be written as $\dot{\xi} = (A-HC) \xi- He$, which is independent of $t_d$, and can be stabilized with the event condition as provided in further analysis in the paper. Without loss of generality, we consider the following assumption for the initialization of the observer states:
\begin{assumption} \label{assum_obs}
	The state estimation error is zero at the instant of time when an attack starts. For instance, if the attack begins at time $t = t_a$, $\xi(t_a) \approx \pmb{0}_{Nm} \iff \hat{x}(t_a) \approx x(t_a)$.
\end{assumption}

Please note that Assumption~\ref{assum_obs} does not imply a prior knowledge of the attack time $t_a$ and only assumes that the observer error is small at the beginning of an attack. Similar assumptions have been made in literature \cite{boem2017distributed}, as such assumptions are generally required for attack detection with residual based attack detection approaches but not mentioned explicitly. We have the following proposition about the boundedness of the signal $B(\hat{u}^a - u^a)$ (whose proof is discussed in Section~\ref{attack_sec}):

\begin{prop} \label{proposition}
Under the proposed event-based observer \eqref{observer}, it holds that $\|B(u^a - \hat{u}^a)\| \leq \varPhi, \forall t \geq t_a$, where $\varPhi = \psi t_d\|B\| \in \mathbb{R}_+$ is a constant.
\end{prop}

\subsubsection{The Controller}
Let us introduce a logic $\delta$ such that $\delta = 0$ indicates ``no attack detection", while $\delta = 1$ indicates an ``attack detection". Based on this, the following switching-based control strategy is proposed:
\begin{equation}\label{switching_control}
u = u_{\delta} \coloneqq (1 - \delta)u^n + \delta u^d,
\end{equation}
where $u^n $ and $u^d$ are the control actions in \textit{normal} (no attack detection) and \textit{defense} (attack detection) modes, respectively. In the absence of an attack, it follows from \cite[Theorem~2.1]{chopra2012output} that the system \eqref{model_complete}, under passivity assumption, achieve output consensus with the following control: 
\begin{equation}\label{control_matrix}
u^n = -K \mathcal{L} y,
\end{equation}
where $\mathcal{L} \coloneqq (\mathcal{L}_s \otimes I_p) \in \mathbb{R}^{Np \times Np}$ is the extended Laplacian, $K \in \mathbb{R}_+$ is a constant gain, and $u^n \in \mathbb{R}^{Np}$ is the stacked input vector under \textit{normal} operation. Under the action of an observer, \eqref{control_matrix} can be written equivalently in terms of the estimated measurement signal $\hat{y}$ as \cite{li2011consensus}:
\begin{equation} \label{control_matrix_observer}
\hat{u}^n = -K \mathcal{L} \hat{y}, 
\end{equation}
which would be used in the subsequent analysis in this paper. 

\subsection{The Problem} \label{Prob}
Consider the linear passive system \eqref{model_complete} under FDI attacks \eqref{attack_ip} with observer dynamics \eqref{observer} equipped with an event-based authentication mechanism as shown in Fig.~\ref{fig:attack_model}. Let the system be governed by a switching-based control protocol \eqref{switching_control}, where the attack estimator operates only in the defense mode and the event-authentication mechanism is in operation for all time. Suppose that the aforementioned assumptions hold. The following problems are addressed in this paper:
\begin{itemize}
\item[{\bf (P1)}] Obtain an event condition for authentication feedback to the observer such that the observer error dynamics \eqref{state_error} is practically stable and the estimation error remains near the origin, i.e., $\lim_{t \to \infty}\|\xi\| \approx 0$.
\item[{\bf (P2)}] Device an FDI attack detection mechanism using passivity inequality \eqref{passivity_inequality} with the quadratic storage function.
\item[{\bf (P3)}] Design the switching control law $u$ in \eqref{switching_control} (effectively $u^d$) such that the system achieves (practical) output consensus, i.e., there exists a small constant $\epsilon_o > 0$ such that $\lim_{t \to \infty}\|y_i - y_j\| \leq \epsilon_o, \ \forall (i, j) \in \mathbb{E}$.
\end{itemize}

\section{Event Condition and Zeno Behavior} \label{mitigation}
In this section, we introduce the event condition and show that the observer states follow the agents' states closely with a small error. To ensure that the observer error dynamics \eqref{state_error} is stable and $\xi$ converges to zero, $V=\xi^T P \xi$ could be a candidate Lyapunov function, where $P$ is a symmetric positive definite matrix satisfying $(A-HC)^T P + P (A-HC) = -Q$ for some symmetric positive definite matrix $Q$. The time-derivative of the Lyapunov function $V$, along the error dynamics \eqref{state_error}, is
\begin{align*} \label{lyap_mid}
	\dot{V} & = \xi^T P \dot{\xi} + \dot{\xi}^T P \xi \\
	& = \xi^T (P(A - HC) + (A - HC)^T P)\xi \\
	& ~~~ + 2 \xi^T P B(u^a - \hat{u}^a) - 2 \xi^T P He \\
	& \leq - \xi^T Q \xi + 2 \|\xi\| \|P B(u^a - \hat{u}^a)\| + 2 \|\xi\| \|P He\|.
\end{align*}
From the inequality\footnote{The proof follows from the fact that the arithmetic mean is greater than or equal to geometric mean for any two positive numbers $cX^2$ and $Y^2/c, c > 0$.} $2 XY \leq c X^2 + ({1}/{c}) Y^2$ for some $c>0$ and $X, Y \in \mathbb{R}$, it follows that $2 \|\xi\| \|P B(u^a - \hat{u}^a)\| \leq  c\|\xi\|^2 {+} ({\|P B(u^a - \hat{u}^a)\|^2}/{c})$ and $2 \|\xi\| \|P He\| \leq c\|\xi\|^2 {+} ({\|P He\|^2}/{c})$. Using these, it can be written that 
\begin{align*}
	\dot{V} & {\leq} - \lambda_{\text{min}} (Q) \|\xi\|^2 {+} 2 c\|\xi\|^2 \\
	& ~~~ + (1/c)({\|P B(u^a {-} \hat{u}^a)\|^2} + {\|P He\|^2})\\
	& {\leq} - \lambda_{\text{min}} (Q)\left(1-\frac{2c}{\lambda_{\text{min}}(Q)}\right) \|\xi\|^2\\
	& ~~~ + (1/c)({\|P B(u^a - \hat{u}^a)\|^2} {+} {\|P He\|^2})\\
	& {\leq} - \frac{\lambda_{\text{min}} (Q)}{\lambda_{\text{max}} (P)}\left(1-\frac{2c}{\lambda_{\text{min}}(Q)}\right) \xi^T P \xi\\
	& ~~~ +  \frac{\|P\|^2}{c}\|B(u^a - \hat{u}^a)\|^2 + \frac{\|PH\|^2}{c} \|e\|^2\\
	& {\leq} - \alpha V(\xi) + \beta \varPhi^2 + \gamma \|e\|^2,
\end{align*}
where $\alpha \coloneqq \frac{\lambda_{\text{min}}(Q)}{\lambda_{\text{max}}(P)}\left(1-\frac{2c}{\lambda_{\text{min}}(Q)}\right) > 0$, $\beta \coloneqq {\|P\|^2}/{c} >0$, $\gamma \coloneqq {\|PH\|^2}/{c}>0$, and $c \in (0, \frac{1}{2}\lambda_{\text{min}}(Q))$ is a design variable. If the event condition is set as $\gamma \|e\|^2 + \beta \varPhi^2 \leq \rho \alpha V(\xi)$ for some $\rho \in (0,1)$, $\dot{V} \leq -\alpha(a-\rho)V$, implying stability of the observer error dynamics. However, non-availability of the entire error vector $\xi$ and requirement to monitor both system and observer states restricts its implementation, as discussed in Section~\ref{Intro}. Therefore, motivated by the dynamic event-triggered scheme \cite{petri2021event}, we introduce an auxiliary variable $\eta$ with dynamics:
\begin{equation} \label{aux}
	\dot{\eta} = -c_1 \eta + c_2 \|e\|^2,
\end{equation}
where $\eta \in \mathbb{R}$, and $c_1>0$ and $c_2 \geq 0$ are design parameters. Clearly, solution of \eqref{aux} is given by $\eta = {\rm e}^{-c_1 t} \eta(0) + \int_{0}^{t} {\rm e}^{-c_1 (t-\tau)} c_2 \|e\|^2 d\tau$, which non-negative for $\eta(0) \geq 0$. Thus, $\eta \geq 0, \forall t \geq 0 \iff \eta(0) \geq 0$.

\begin{thm} \label{event_thm}
	The observer error dynamics \eqref{state_error} is stable and the error $\xi$ converges to a small ball around the origin under the following dynamic event condition:
	\begin{equation} \label{event_cond}
		({\gamma +d c_2}) \|e\|^2 =  (d c_1 \eta + \varOmega),
	\end{equation}
	where $\varOmega = \varepsilon - \beta \varPhi^2 > 0$, $\varepsilon \in \mathbb{R}_+$ and $d \in \mathbb{R}_+$ are arbitrary design constants ($\beta, \gamma$ are defined as above).
\end{thm}
\begin{proof}
	Consider the composite Lyapunov function $U = V(\xi) + d \eta$, whose time-derivative, along \eqref{state_error} and \eqref{aux}, is $\dot{U}	 {=} \dot{V} + d \dot{\eta}  \leq - \alpha V(\xi) + \gamma \|e\|^2 + \beta \varPhi^2 - d c_1 \eta + d c_2 \|e\|^2  = - \alpha V(\xi) + (\gamma + dc_2) \|e\|^2 + \beta \varPhi^2 - d c_1 \eta$. Under the event condition \eqref{event_cond} and substituting for $\varOmega$, yields $\dot{U} \leq - \alpha V(\xi) + \varepsilon$. Now, it can be concluded that the dynamics \eqref{state_error} is practically stable and the estimation error $\xi$ converges to a small ball near the origin having its radius dependent on $\varepsilon$.
\end{proof}

Next, we prove the exclusion of Zeno behavior by showing that there exists a positive minimum inter event time (MIET). Before proceeding, we consider that the system's output has bounded derivative, i.e., $\|\dot{y}\| = \|C\dot{x}\| = \|CAx + CBu\| \leq \sigma, \ \forall t \geq 0$, where $\sigma$ is arbitrary large positive constant. 

\begin{thm} \label{MIET_thm}
	Consider error dynamics \eqref{state_error}. Under the event condition \eqref{event_cond}, the system \eqref{state_error} has positive MIET, $\tau = \frac{1}{2\sigma} \sqrt{\frac{\varOmega}{\gamma +d c_2}}$ and does not exhibit Zeno behavior.
\end{thm}
\begin{proof}
	Differentiating $\|e\|^2$ for $t \in [t_k, t_{k+1})$ we get $\frac{d}{dt}\|e\|^2 = 2 e^T \dot{e} = 2 e^T (\dot{\bar{y}} - \dot{y})$ $= -2 e^T C(Ax + Bu)$ $\leq 2 \|e\| \|CAx + CBu\|$ $\leq 2 \sigma \|e\|$. For MIET $(\tau)$, the event triggers when the right side of the equality in \eqref{event_cond} is minimum, which is achieved when $\eta = 0$ (since $\eta \geq 0 \forall t \geq 0$, see \eqref{aux}). The value of error at MIET $(\tau)$ is 
	\begin{eqnarray}\label{next_event}
		\|e(t_k + \tau)\| = \min_\eta\sqrt{\frac{dc_1\eta + \varOmega}{\gamma +d c_2}} = \sqrt{\frac{\varOmega}{\gamma +d c_2}}, 
	\end{eqnarray}
	and using \eqref{next_event} in the expression of $\frac{d}{dt} \|e\|^2$, we get
	\begin{equation}\label{event_derivative}
		\frac{d}{dt} \|e\|^2 \leq 2\sigma\sqrt{\frac{\varOmega}{\gamma +d c_2}}.
	\end{equation}
	Integrating \eqref{event_derivative} for $t \in [t_k, t_k + \tau)$, we get $\|e(t_k + \tau)\|^2 - \|e(t_k)\|^2 = 2\sigma (t_k + \tau - t_k)\sqrt{{\varOmega}/(\gamma +d c_2)}$. Since $\|e(t_k)\| = 0, \forall k \in \mathbb{Z}_+$ and substituting $\|e(t_k + \tau)\|$ from \eqref{next_event}, we get $\tau = ({1}/{2\sigma}) \sqrt{{\varOmega}/(\gamma +d c_2)}$.
\end{proof}

\section{Attack Detection, Mitigation and Closed-Loop System Stability} \label{attack_sec}
\subsection{Passivity-Based Attack Detection}
The proposed passivity-based attack detector relies on the verification of passivity inequality \eqref{passivity_inequality} for the measurable signals in Fig.~\ref{fig:attack_model}. Note that the agents' output $y$, controllers' output $u$, and observers' state $\hat{x}$ and the output $\hat{y} = C\hat{x}$ are measurable. Note that the original input signal $u$ is measurable and the corrupted input $u^c$, applied to the agent, is not measurable. However, the output $y$ is available only at events of cryptographic authentication. This suggests verification of \eqref{passivity_inequality} for the input-output pair $(u, \hat{y})$ across the observer (i.e., points E and F in Fig.~\ref{fig:attack_model}) with the quadratic storage function $S(\hat{x}) = (1/2)\hat{x}^T \hat{x}$ associated with the observer's states $\hat{x}$. It is straightforward that \eqref{passivity_inequality} holds for appropriately designed observer, across points E and F if there is no attack on the system \cite{qin2009observer}, while this might not be true in case of an actuator attack. We leverage this fact in detecting the presence of an attack $u^a$. We classify the attack signal as detectable and undetectable as follows:
\begin{definition} \label{def_detectable}
An attack $u^a$ which does not satisfy (resp., satisfy) passivity inequality \eqref{passivity_inequality} with respect to the points E and F in Fig.~\ref{fig:attack_model} with quadratic storage function $S(\hat{x}) = (1/2)\hat{x}^T \hat{x}$ is detectable (resp., undetectable).
\end{definition}

Note that it is sufficient to check \eqref{passivity_inequality} for the detection of an attack on the system. However, we also provide the following theorem to establish a connection between the system and network properties (in terms of the matrices $A, B, C, H$, and $\mathcal{L}$) under the attack signal $u_a$ on the system, however, it is not necessary for the purpose of implementation $-$ verifying only \eqref{passivity_inequality} is sufficient. 

\begin{thm} \label{passive_detect_thm}
	Let $\mathcal{M} \coloneqq (K(B\mathcal{L} - C^T \mathcal{L}^T) + H)C - A$ be a matrix of order $Nm \times Nm$. If there exists an attack signal $u^a$ such that
	\begin{equation} \label{passive_detect_eqn}
	\hat{x}^T \mathcal{M} \hat{x} < \hat{x}^T HC \int_{0}^{t_k} (Ax + B\hat{u}^n +B u^a) d\tau,
	\end{equation}
	where $t_k$ is the latest event time, then the passivity inequality \eqref{passivity_inequality} is not satisfied across points E and F. Consequently, an attack $u^a$ is detected in accordance with Definition~\ref{def_detectable}.
\end{thm}
\begin{proof}
	We prove it by contradiction. Assume that \eqref{passivity_inequality} is satisfied across points E and F with $S(\hat{x}) = (1/2)\hat{x}^T \hat{x}$. This implies that $u^T \hat{y} \geq \dot{S} = \hat{x}^T \dot{\hat{x}}$. Since the system is operating in normal mode, we have $u = \hat{u}^n$ as $\delta = 0$ in \eqref{switching_control}. Substituting $\dot{\hat{x}}$ and $\hat{u}^n$ from \eqref{observer} and \eqref{control_matrix_observer} results in $-K \hat{x}^T C^T \mathcal{L}^T C \hat{x} \geq  \hat{x}^T (A\hat{x} + B\hat{u}^n + H (y - \hat{y} + e)) \implies -K \hat{x}^T C^T \mathcal{L}^T C \hat{x} \geq  \hat{x}^T (A\hat{x} - BK\mathcal{L}C\hat{x} + HC (x - \hat{x} + \bar{x} - x))$, where $\bar{x} = x(t_k)$. Rearranging, we get $\hat{x}^T ((K(B\mathcal{L} -C^T \mathcal{L}^T) + H)C - A)\hat{x} \geq \hat{x}^THC \bar{x}$. Replacing $(K(B\mathcal{L} - C^T \mathcal{L}^T) + H)C - A$ by $\mathcal{M}$ as defined in the statement of theorem, the preceding inequality becomes $\hat{x}^T \mathcal{M} \hat{x} \geq \hat{x}^T HC \bar{x}$. Now, substituting $\bar{x} = \int_{0}^{t_k} \dot{x}(\tau) d \tau$, where $\dot{x}(\tau)$ is given by \eqref{model_complete}, we get $\hat{x}^T \mathcal{M} \hat{x} \geq \hat{x}^T HC \int_{0}^{t_k} (Ax + B\hat{u}^n +B u^a) d\tau$. Now, it can be concluded that any $u^a$ which violates the preceding inequality (i.e., satisfies \eqref{passive_detect_eqn}) does not satisfy \eqref{passivity_inequality} across points E and F in Fig.~\ref{fig:attack_model}, and hence, are detectable. 
\end{proof}

\subsection{Attack Mitigation}
Once the attack is detected, the controller \eqref{switching_control} switches to the defense mode, i.e., $\delta = 1$ and $u = u^d$. Relying on the estimated attack signal $\hat{u}^a$ in \eqref{attack_estimate}, we propose the control in defense mode as
\begin{equation}\label{control_defense_mode}
u^d = \hat{u}^n - \hat{u}^a,
\end{equation}
where $\hat{u}^n$ is given by \eqref{control_matrix_observer}. As a result, the compromised control input to the actuators becomes
\begin{equation} \label{control_under_attack}
	u^c = u^d + u^a = \hat{u}^n - \hat{u}^a + u^a.
\end{equation}
We are now ready to prove Proposition~\ref{proposition}. 

\begin{proof}[\textbf{Proof of Proposition~\ref{proposition}}]
	Left multiplying \eqref{attack_estimate} with $B$ on both sides, yields
	\begin{equation}
		B \hat{u}^a = B B^\dagger (\dot{\hat{x}} - A\hat{x} - Bu^n + B\hat{u}^a(t - t_d)).
	\end{equation}
	From \eqref{state_error}, substituting $\hat{x} = x - \xi$ and $\dot{\hat{x}} = \dot{x} - \dot{\xi}$, we obtain $B \hat{u}^a = B B^\dagger(\dot{x} - \dot{\xi} - A(x - \xi) - B\hat{u}^n + B\hat{u}^a (t - t_d))$. Substituting for $\dot{x}$ and $u^c$ from \eqref{model_complete} and \eqref{control_under_attack}, respectively, and using the property $B B^\dagger B = B$ (as stated in notations), we get $B \hat{u}^a= B B^\dagger (Ax + B\hat{u}^n - B\hat{u}^a + Bu^a - Ax -B\hat{u}^n+ B\hat{u}^a (t - t_d) - \dot{\xi} + A\xi)$ $= Bu^a + B\hat{u}^a (t-t_d) - B\hat{u}^a + \nu$, where $\nu = B B^\dagger (-\dot{\xi} + A\xi)$. Rearranging, we have $B (\hat{u}^a - u^a) = - B(\hat{u}^a - \hat{u}^a (t - t_d)) + \nu \implies \|B (\hat{u}^a - u^a)\| \leq \|B(\hat{u}^a - \hat{u}^a (t - t_d))\| + \|\nu\|$. Using \eqref{lipsch}, we get $\|B (\hat{u}^a - u^a)\| \leq \psi t_d\|B\| + \|\nu\|$. According to Assumption~\ref{assum_obs} and Theorem~\ref{event_thm}, $\xi \to 0$ as $t \to t_a$, this implies that $\|\nu\| \to 0$ as $t \to t_a$. Now, it follows that there exists a constant $\varPhi = \psi t_d\|B\|> 0$ such that $\|B (\hat{u}^a - u^a)\| \leq \varPhi, \ \forall t \geq t_a$, proving our claim.
\end{proof}

\begin{remark}\label{attack_estimate_error_bound}
Multiplying by $\|B^\dagger\|$ on both sides of the inequality $\|B (\hat{u}^a - u^a)\| \leq \varPhi$, we have $\|B^\dagger\|\|B (\hat{u}^a - u^a)\| \leq \varPhi\|B^\dagger\| \implies \|B^\dagger B (\hat{u}^a - u^a)\| \leq \|B^\dagger\|\|B (\hat{u}^a - u^a)\| \leq \varPhi\|B^\dagger\|$. Since $B^\dagger B = I$ for the full column rank matrix $B$ (see {\bf (A2)}) and $\varPhi = \psi t_d\|B\|$, it follows from the preceding relation that $\|\hat{u}^a - u^a\| \leq \psi t_d, \ \forall t \geq t_a$.  Clearly, $\|\hat{u}^a - u^a\| \approx 0, \forall t \geq t_a$ for sufficiently small reaction time $t_d$.
\end{remark}

Note that there seems to be an inter-dependency among Proposition~\ref{proposition} and Theorem~\ref{event_thm}. However, under Assumption~\ref{assum_obs}, Theorem~\ref{event_thm} holds independently for $t < t_a$, as the attack estimator is ``OFF" (i.e., $\hat{u}^a = \pmb{0}_{Np}$) in the absence of an attack (i.e., $u^a=\pmb{0}_{Np}$), and can be considered as initialization to our problem. On the other hand, in the presence of an attack (i.e., $u^a \neq \pmb{0}_{Np}$), we make sure with the aid of Assumption~\ref{attack_assum} that when an attack starts at $t = t_a$, the state observer has a small error for accurate detection and estimation of the attack signal (similar assumption is also made in \cite{boem2017distributed}). 

\subsection{Closed-Loop System Stability} \label{stability_sec}
In the below theorem, we show that the difference among agents' outputs remain bounded during attack and the closed-loop system achieves practical consensus.  
\begin{thm} \label{stability_thm}
Consider the closed-loop system as shown in Fig.~\ref{fig:attack_model} with agent, observer and attack estimator models \eqref{model_complete}, \eqref{observer} and \eqref{attack_estimate}, respectively. Let the system be governed by the switching controller \eqref{switching_control} where $u^n$ and $u^d$ are defined in \eqref{control_matrix_observer} and \eqref{control_defense_mode}, respectively. Then, the system achieves (practical) output consensus, i.e., there exists a small constant $\epsilon_o > 0$ such that $\lim_{t \to \infty}\|y_i - y_j\| \leq \epsilon_o, \forall (i, j) \in \mathbb{E}$. Additionally, if the communication topology is balanced, then $\epsilon_o$ is given by $\epsilon_o = \sqrt{\bar{\omega}/K}$, where $\bar{\omega}$ depends on $\|\xi\|$ and $\|\hat{u}^a - u^a\|$. 
\end{thm}

\begin{proof}
Under no attack condition, the convergence directly follows from \cite{chopra2012output}, as $u = \hat{u}^n$ (see \eqref{control_matrix_observer}). Our main goal here is to prove convergence in case when the controller operates in the defense mode under attack, i.e., $u = u^d$, given by \eqref{control_defense_mode}. Since the agents \eqref{model_complete} are passive, Lemma~\ref{passive} assures existence of a candidate Lyapunov function $V_s(x)$ such that its time-derivative along agent dynamics \eqref{model_complete} satisfies $\dot{V}_s = ({\partial V_s}/{\partial x})^T (Ax + Bu^c) = -W(x) + y^T u^c = -W(x) + y^T \hat{u}^n - y^T (\hat{u}^a - u^a)$, where $W(x)>0, \forall x \neq \pmb{0}_{Nn}$. Substituting for $\hat{u}^n$ from \eqref{control_matrix_observer} and $\hat{x} = x -\xi$, we obtain
\begin{align*}
\dot{V}_s & = -W - K y^T \mathcal{L} \hat{y} - y^T (\hat{u}^a - u^a)\\
& = -W - K y^T \mathcal{L} C \hat{x} - y^T (\hat{u}^a - u^a)\\
& = -W - K y^T \mathcal{L} C (x - \xi) - y^T (\hat{u}^a - u^a)\\
& = -W - K y^T \mathcal{L} y {+} K y^T \mathcal{L} C \xi {-} y^T (\hat{u}^a {-} u^a)\\
& = -W - K y^T \mathcal{L} y + \omega,
\end{align*}
where $\omega = y^T(K \mathcal{L} C \xi - (\hat{u}^a - u^a))$. If $\omega \equiv 0$, the trajectories converges to the set $\Gamma = \{x \in \mathbb{R}^{Nm} \ | \ \dot{V}_s \equiv 0\}$. The set $\Gamma$ is characterized by all the trajectories such that $\{W(x) \equiv 0, y \equiv \mu \pmb{1}_{Np}, \mu \in \mathbb{R}\}$. If $\omega \neq 0$, it can be written that $\dot{V}_s \leq -W - K y^T \mathcal{L} y + |\omega|$. Following Proposition~\ref{proposition} and Remark~\ref{attack_estimate_error_bound}, it can be concluded that $|\omega| \leq \bar{\omega}$ for some $\bar{\omega}$, which is a small positive constant for small values of $\|\xi\|$ and $\|\hat{u}^a {-} u^a\|$. Consequently, $\dot{V}_s \leq -W(x) - K y^T \mathcal{L} y + \bar{\omega}$, and the convergence follows in the sense of practical stability.

To prove the second statement, we analyze the time-derivative $\dot{V}_s$. Clearly, $\dot{V}_s \leq 0$, if $\bar{\omega} \leq W(x) + K y^T \mathcal{L} y$. Since the preceding inequality must hold even in the worst-case scenario where $W(x) \equiv 0$, the condition $K y^T \mathcal{L} y \geq \bar{\omega}$ must be satisfied. In this situation, as $\dot{V}_s \leq 0$ outside the region where $K y^T \mathcal{L} y \geq \bar{\omega}$, it follows from \cite[Lemmea~4.6]{khalil2002nonlinear} that all the solution trajectories fall within the region where $K y^T \mathcal{L} y \leq \bar{\omega}$, as $t \to \infty$. For balanced and strongly connected digraphs (i.e., $L = L^T$), the preceding inequality can be written as $\sum_{(i, j) \in \mathbb{E}} \|y_i - y_j\|^2 \leq \bar{\omega}/{K} \implies \|y_i - y_j\| \leq \sqrt{\bar{\omega}/{K}}$ for each $(i, j) \in \mathbb{E}$. As $\bar{\omega}$ relies on $\|\xi\|$ and $\|\hat{u}^a {-} u^a\|$, small values of  $\|\xi\|$ and the reaction time $t_d$ (see Remark~\ref{attack_estimate_error_bound}) will lead to small value of $\bar{\omega}$, and hence, $\epsilon_o$. Alternatively, by increasing $K$, $\epsilon_o$ can be made small.
\end{proof}

\begin{figure}[t]
	\centering
	\subfloat[][Network Topology]{\includegraphics[width = 0.17\textwidth]{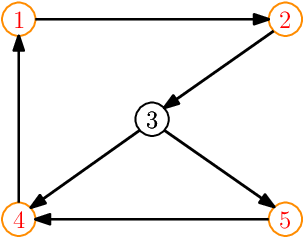} \label{fig:graph}}
	\subfloat[][Laplacian]{$\mathcal{L}_s =  \begin{bmatrix}
			1 &  0 &  0 & -1 &  0\\
			-1 &  1 &  0 &  0 &  0\\
			0 & -1 &  1 &  0 &  0\\
			0 &  0 & -1 &  2 & -1\\
			0 &  0 & -1 &  0 &  1
		\end{bmatrix}$
		\label{laplacian}}
	\caption{Communication Topology.}
	\label{fig:topology}
\end{figure}

\section{Simulation Example} \label{Sim}
Consider $N=5$ agents interacting according to a communication graph as shown in Fig.~\ref{fig:topology}. Agent's initial states are taken randomly in the interval $[-20, 25]$. The FDI attack on actuator is of the form $u^a_i = a_i \sin{\omega_i t}$, where $a_i, \omega_i$ are chosen randomly in the interval $[10, 20]$ and $[0, 10\pi]$, respectively, for $i=1,2,4,5$, and $a_3 = 0$, i.e., no attack on agent $3$. Note that the attack information is given only for the purpose of illustration and is unavailable to the controller. We consider that the attack remains active for $t \in [2, 5] s$. We discuss three case studies as follows:

\subsection{Passive Agents with Real Poles} \label{passive_real_poles}
Consider heterogeneous agents with the following governing matrices in \eqref{model_complete}:
\begin{equation*}
	A_i=i \begin{bmatrix} -6 & i & 0.5i\\ i & -9 & 0.4i\\ 0.5i & 0.4i & -9 \end{bmatrix}, ~ B_i=i \begin{bmatrix} 0\\ 0\\ 0.5 \end{bmatrix}, ~ C_i=\frac{1}{i} \begin{bmatrix} 2 \\ 0 \\ 0 \end{bmatrix}^T,
\end{equation*}
respectively, where $i=1,\ldots,5$. It can be easily verified that the agents are passive as the associated transfer function is positive real (please refer to positive real lemma \cite[Lemma~6.4]{khalil2002nonlinear}). Further, the preconditions \textbf{(A1)} and \textbf{(A2)} are trivially satisfied. 

\begin{figure}[t]
	\centering
	\begin{subfigure}{0.24\textwidth}
		\centering
		\includegraphics[width=\textwidth]{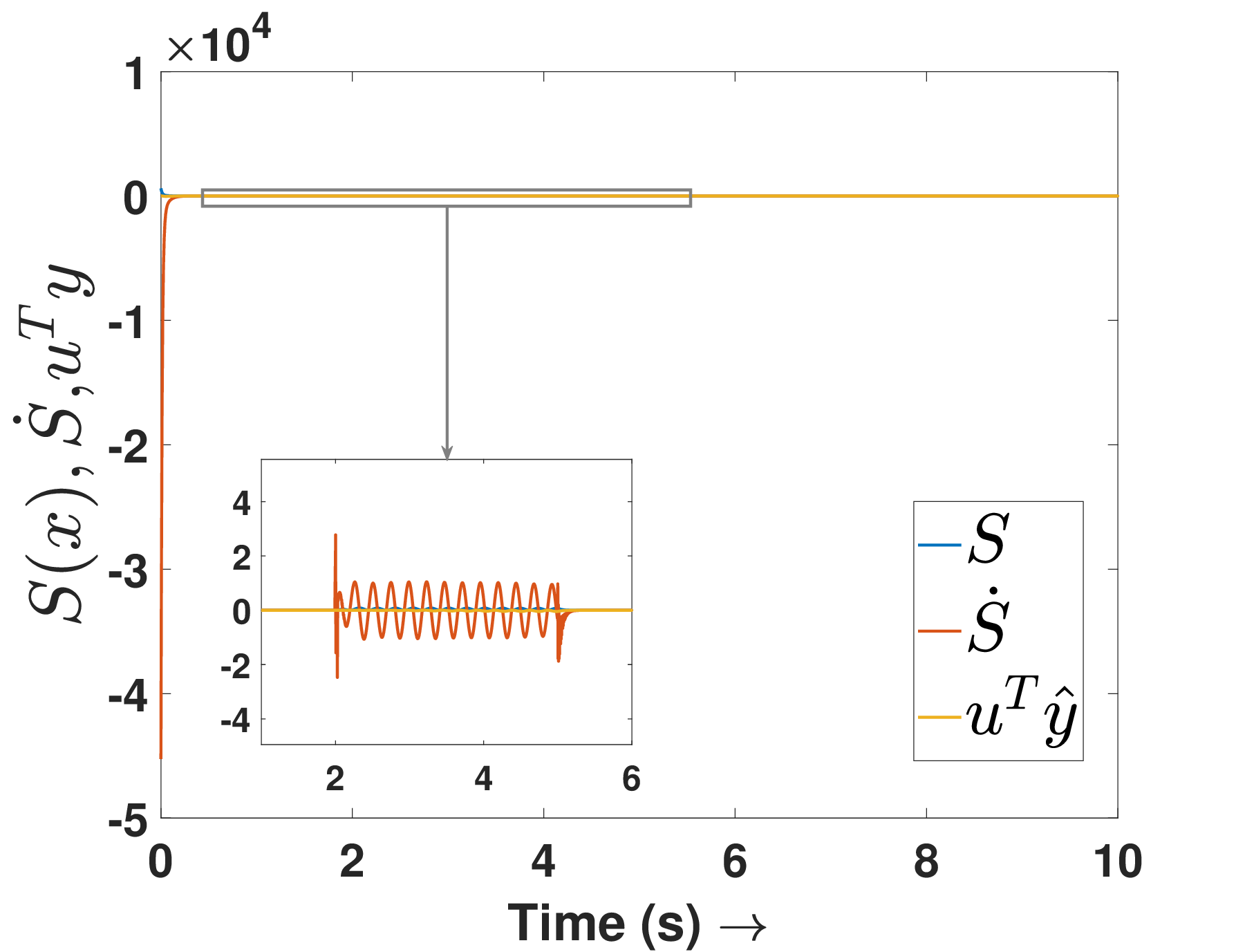}
		\caption{Passivity for agent 4}
		\label{fig:Pass}
	\end{subfigure}
	\hfill \hspace{-15pt}
	\begin{subfigure}{0.24\textwidth}
		\centering
		\includegraphics[width=\textwidth]{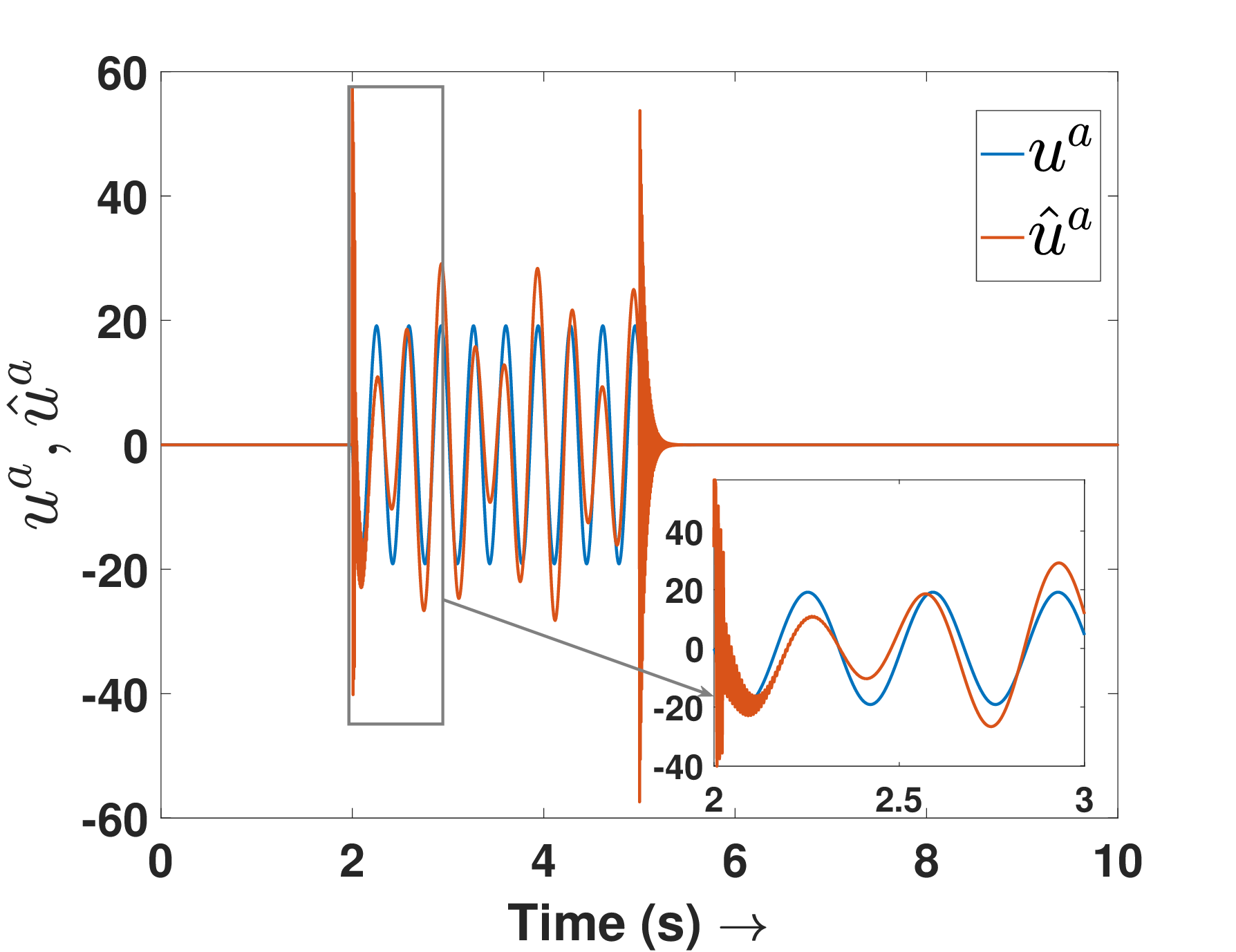}
		\caption{$u^a_4$, $\hat{u}^a_4$}
		\label{fig:Attack}
	\end{subfigure}
	\caption{Passivity inequality and attack signals (Case~A).}
\end{figure}

\begin{figure}[t]
	\centering
	\begin{subfigure}[b]{0.24\textwidth}
		\centering
		\includegraphics[width=\textwidth]{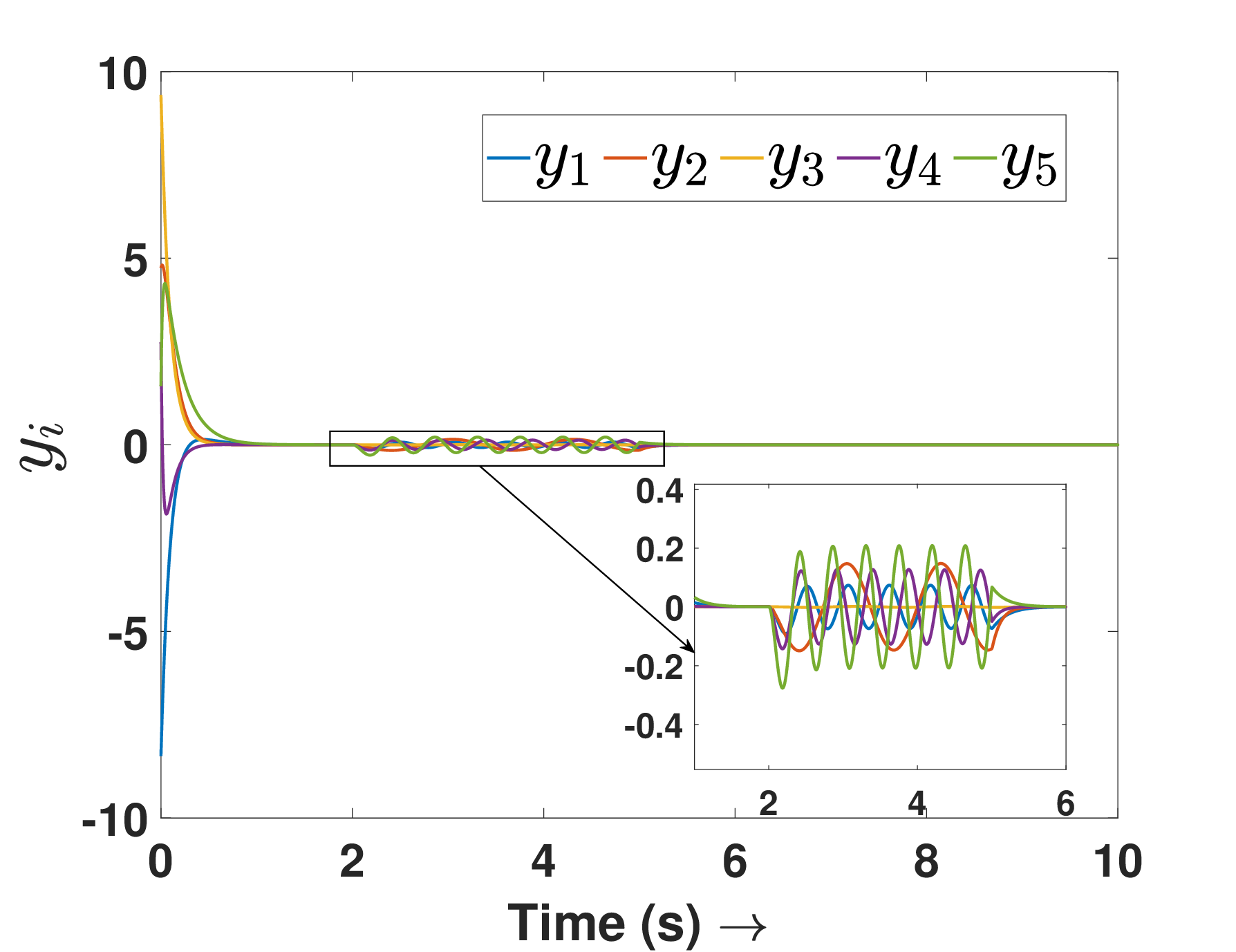}
		\caption{$y_i, \forall i$}
		\label{fig:Output}
	\end{subfigure}
	\hfill \hspace{-15pt}
	\begin{subfigure}[b]{0.24\textwidth}
		\centering
		\includegraphics[width=\textwidth]{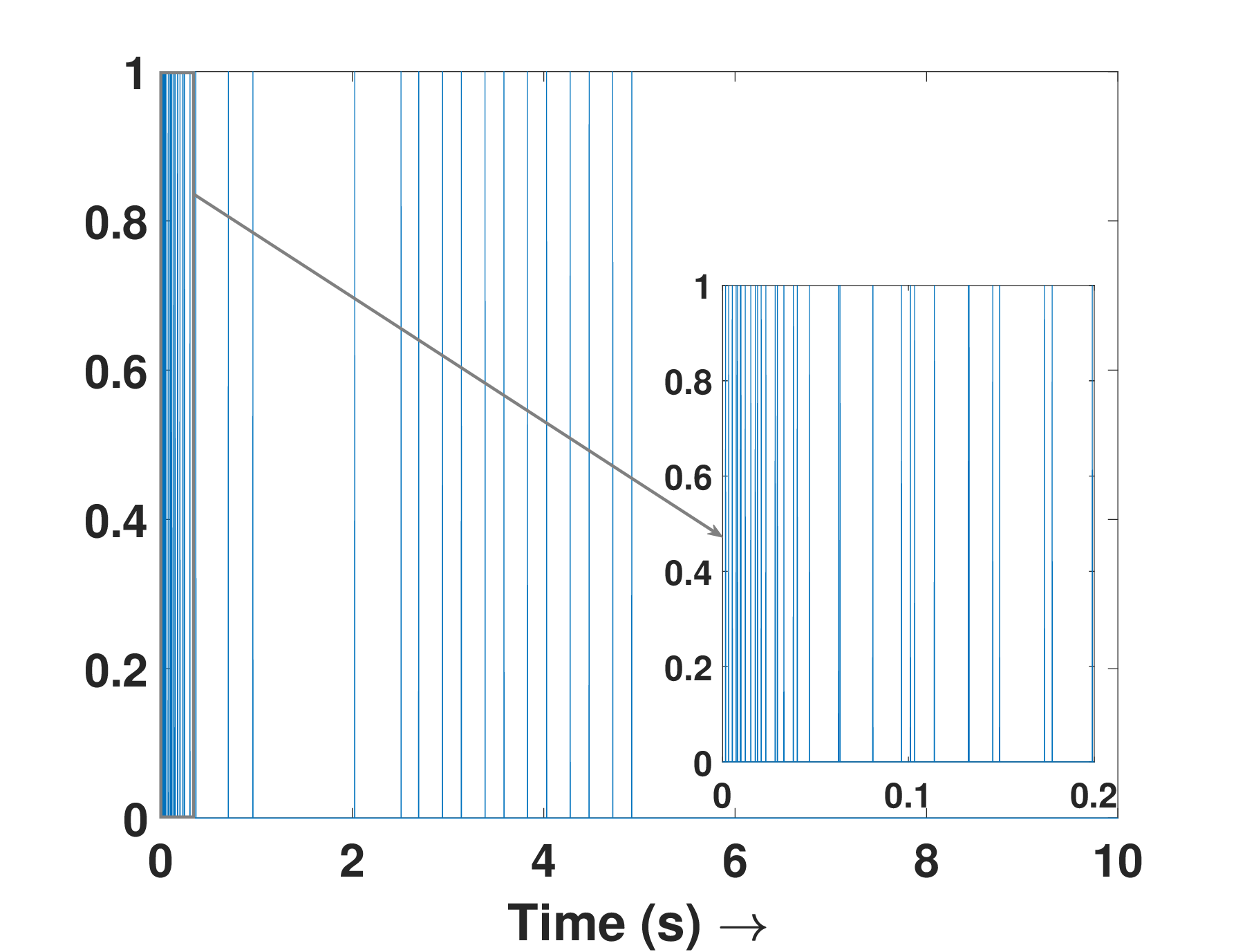}
		\caption{Events}
		\label{fig:Events}
	\end{subfigure}
	\caption{Agents' output and events for observer (Case~A).}
\end{figure}

\begin{figure}[t]
	\centering
	\begin{subfigure}{0.24\textwidth}
		\centering
		\includegraphics[width=\textwidth]{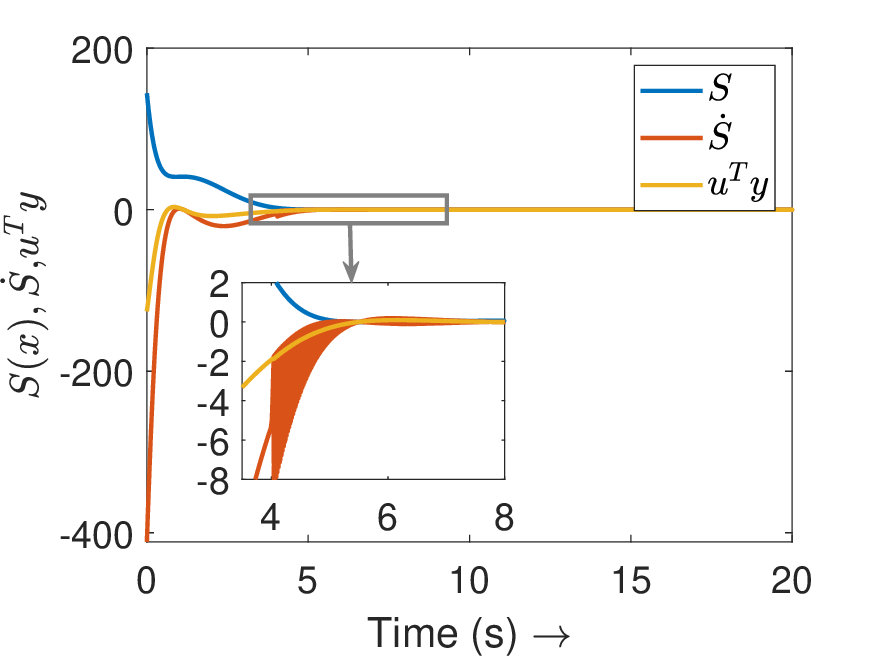}
		\caption{Passivity for agent 2}
		\label{fig:POPass}
	\end{subfigure}
	\hfill \hspace{-15pt}
	\begin{subfigure}{0.24\textwidth}
		\centering
		\includegraphics[width=\textwidth]{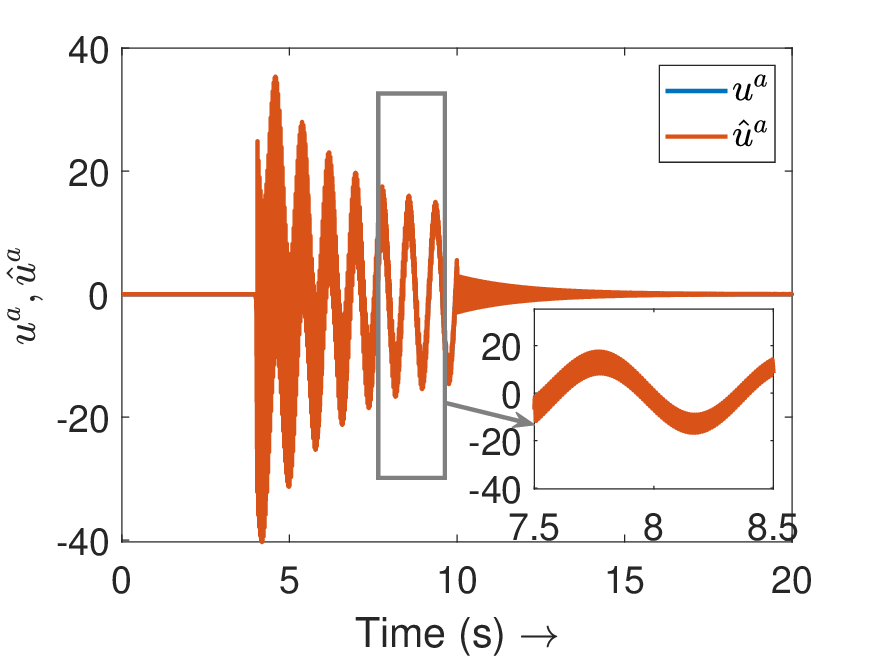}
		\caption{$u^a_2$, $\hat{u}^a_2$}
		\label{fig:POAttack}
	\end{subfigure}
	\caption{Passivity inequality and attack signals (Case~B).}
	\label{fig:PO}
\end{figure}

\begin{figure}[t]
	\centering
	\begin{subfigure}[b]{0.24\textwidth}
		\centering
		\includegraphics[width=\textwidth]{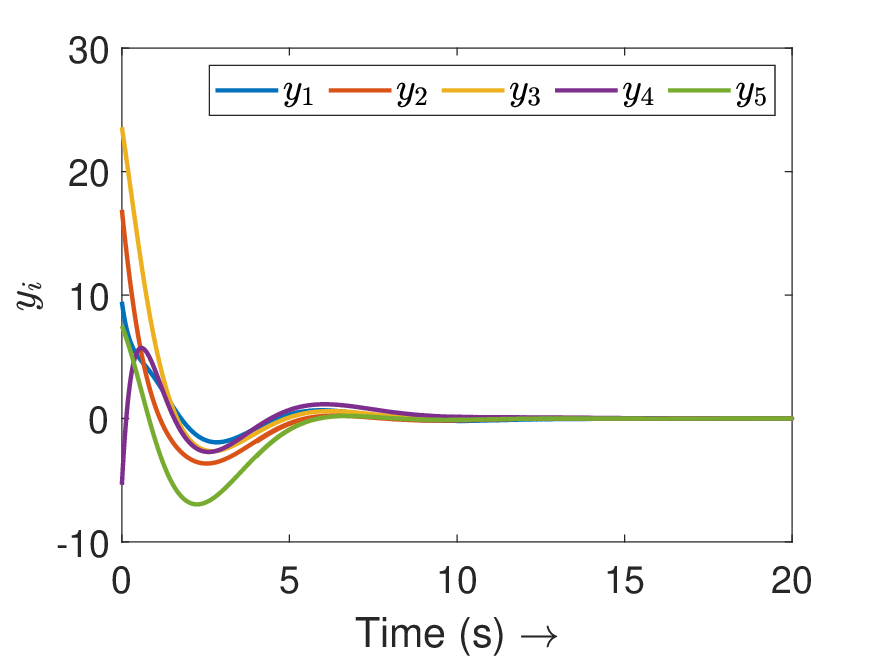}
		\caption{$y_{i}, \forall i$}
		\label{fig:POx2}
	\end{subfigure}
	\hfill \hspace{-15pt}
	\begin{subfigure}[b]{0.24\textwidth}
		\centering
		\includegraphics[width=\textwidth]{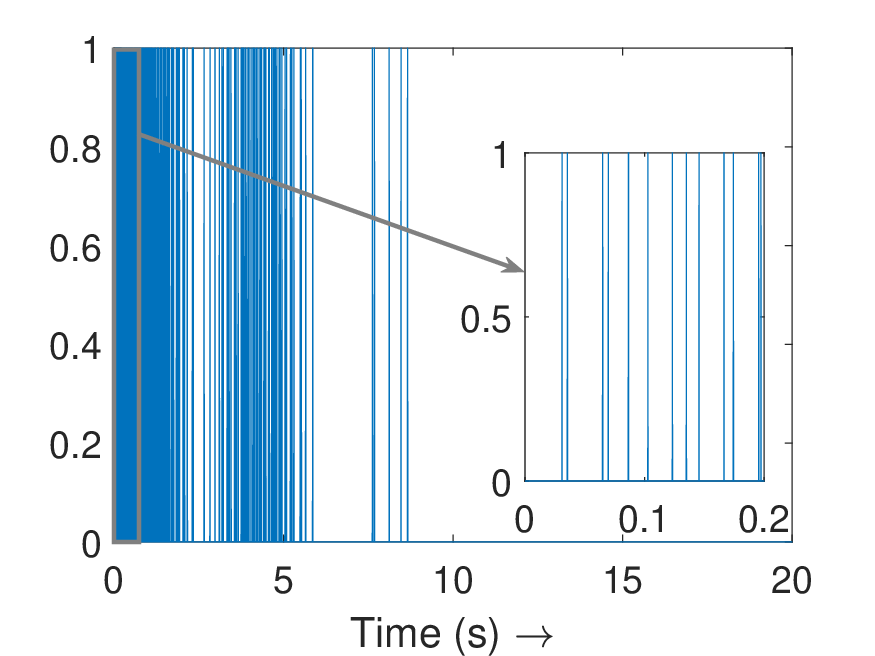}
		\caption{Events for Observer}
		\label{fig:POEvents}
	\end{subfigure}
	\caption{Oscillator's output and events for observer (Case~B).}
	\label{fig:POx}
\end{figure}

\begin{figure}[t]
	\centering
	\begin{subfigure}{0.24\textwidth}
		\centering
		\includegraphics[width=\textwidth]{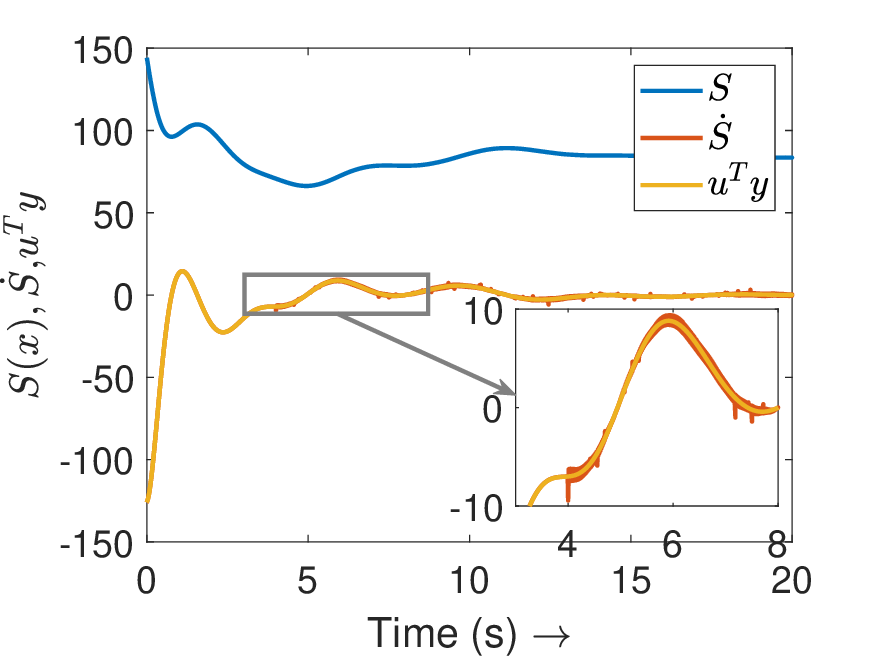}
		\caption{Passivity for agent 2}
		\label{fig:PO2Pass}
	\end{subfigure}
	\hfill \hspace{-15pt}
	\begin{subfigure}{0.24\textwidth}
		\centering
		\includegraphics[width=\textwidth]{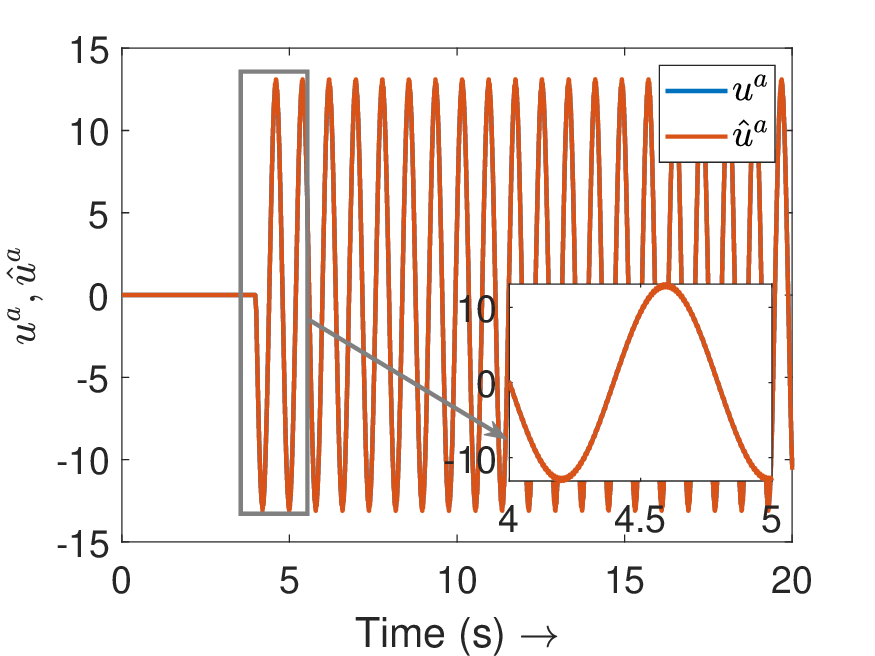}
		\caption{$u^a_2$, $\hat{u}^a_2$}
		\label{fig:PO2Attack}
	\end{subfigure}
	\caption{Passivity inequality and attack signals (Case~C).}
	\label{fig:PO2}
\end{figure}

\begin{figure}[t]
	\centering
	\begin{subfigure}[b]{0.24\textwidth}
		\centering
		\includegraphics[width=\textwidth]{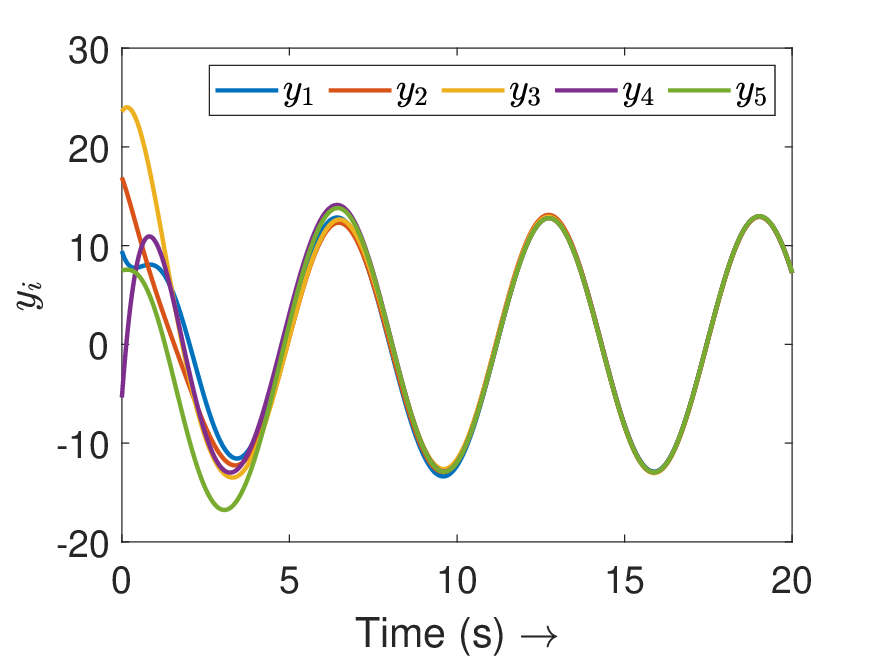}
		\caption{$y_{i}, \forall i$}
		\label{fig:PO2x2}
	\end{subfigure}
	\hfill \hspace{-15pt}
	\begin{subfigure}[b]{0.24\textwidth}
		\centering
		\includegraphics[width=\textwidth]{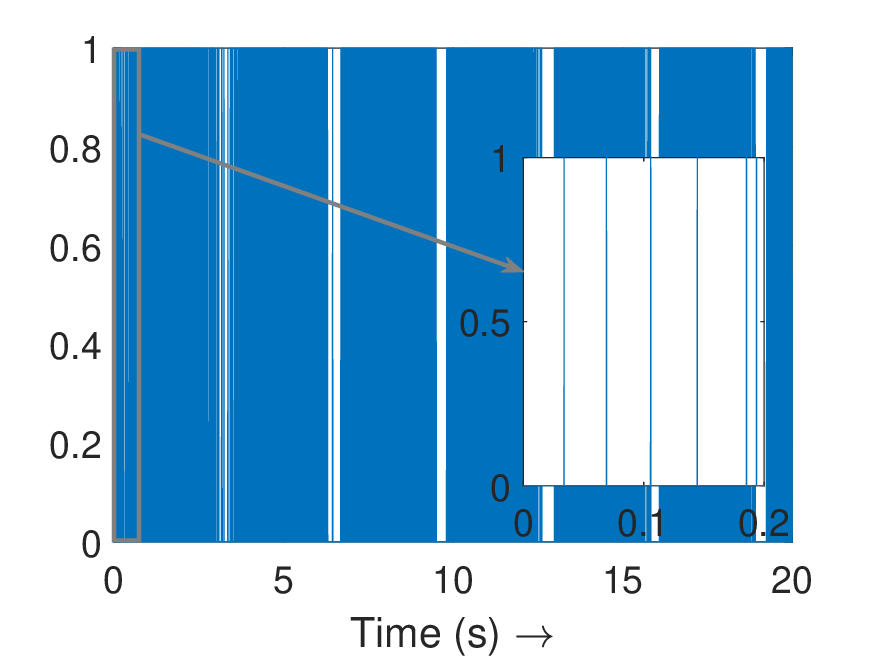}
		\caption{Events for Observer}
		\label{fig:PO2Events}
	\end{subfigure}
	\caption{Oscillator's output and events for observer (Case~C).}
	\label{fig:PO2x}
\end{figure}

Fig.~\ref{fig:Pass} depicts the plots for storage function $S(\hat{x})$, change in the stored energy $\dot{S}(\hat{x})$, and the supplied energy $u^T \hat{y}$ for agent 4. It can be observed that the passivity is lost at $t = 2$s as soon as the attack begins. Consequently, the attack is detected and the controller switches to defense mode. Fig.~\ref{fig:Attack} shows the attack estimation and the actual attack signals $-$ it is clearly visible that their difference is finite. The plots for other agents are similar and omitted for brevity. Fig.~\ref{fig:Output} shows that the output (of all the agents) remains bounded even under attack, as established in Theorem~\ref{stability_thm}. Finally, Fig.~\ref{fig:Events} depicts the events generated for the observer and shows the absence of Zeno behavior. Because of the fast-moving system states, there is a higher density of events initially and during the attack period, and a lower density when the states do not change.

\subsection{Passive Agents with Complex Poles} \label{passive_complex_poles}
Consider the network of homogeneous harmonic oscillators \cite{xia2016synchronization}, characterized by the following matrices in \eqref{model_complete}:
\begin{equation*}
	A_i = \begin{bmatrix}
		0 & 1\\
		-1 & 1
	\end{bmatrix}, \quad
	B_i = \begin{bmatrix}
		0\\
		1
	\end{bmatrix}, \quad
	C_i = \begin{bmatrix}
		0 & 1
	\end{bmatrix}, \ \forall i,
\end{equation*}
is a passive system according to the positive real lemma \cite[Section 6.3]{khalil2002nonlinear} and satisfy the preconditions \textbf{(A1)} and \textbf{(A2)}. Figs.~\ref{fig:PO} and \ref{fig:POx} show the simulation results for this case where conclusions, similar to the previous case, can be drawn. It is observed that the presence of complex poles degrades the attack estimation performance of the system due to the presence of oscillations in the system output.

\subsection{Passive Agents with Imaginary Poles} \label{passive_imaginary_poles}
We consider the network of homogeneous harmonic oscillators \cite{xia2016synchronization}, having the following matrices:
\begin{equation*}
	A_i = \begin{bmatrix}
		0 & 1\\
		-1 & 0
	\end{bmatrix}, \quad
	B_i = \begin{bmatrix}
		0\\
		1
	\end{bmatrix}, \quad
	C_i = \begin{bmatrix}
		0 & 1
	\end{bmatrix}, \forall i.
\end{equation*}
Again, the agents are passive and satisfy the preconditions \textbf{(A1)} and \textbf{(A2)}. Figs.~\ref{fig:PO2} and \ref{fig:PO2x} show the simulation results for this case. It can be noted that i) the presence of imaginary poles degrades the attack estimation performance of the system, and ii) the number of events is relatively higher, as compared to earlier cases, which is expected due to the oscillating nature of output \cite{nowzari2019event}. 

\section{Conclusions} \label{Conclusion}
The problem of output consensus for networked linear passive agents under FDI actuator attacks is examined in this paper. To reduce the computation and communication overhead, an event-triggered observer-based switching controller was proposed in conjunction with cryptographic authentication for estimating the non-measurable states. It was shown that the proposed event condition does not result in Zeno behavior and has a positive MIET. Relying on the measurements from the observer, a passivity-based attack detection approach was presented. It was shown that the error between the estimated and the actual attack signal is norm-bounded during the controller's functioning in the defense mode. Finally, the system was shown to achieve practical output consensus.

\begin{figure}[t]
	\centering
	\begin{subfigure}[b]{0.24\textwidth}
		\centering
		\includegraphics[width=1\textwidth]{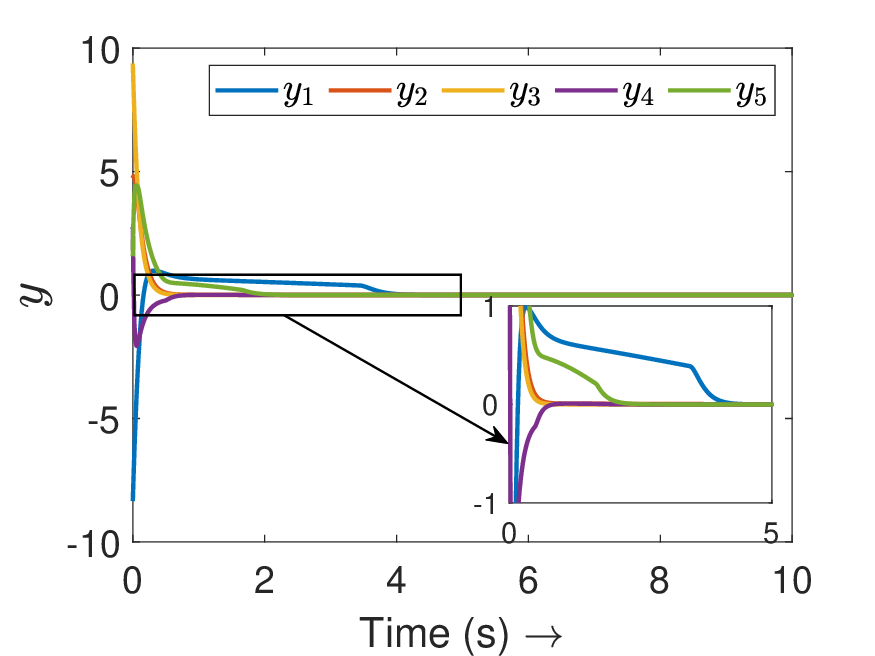}
		\caption{$y_{i}, \forall i$}
		\label{fig:Passive_consensus}
	\end{subfigure}
	\hfill \hspace{-15pt}
	\begin{subfigure}[b]{0.24\textwidth}
		\centering
		\includegraphics[width=1\textwidth]{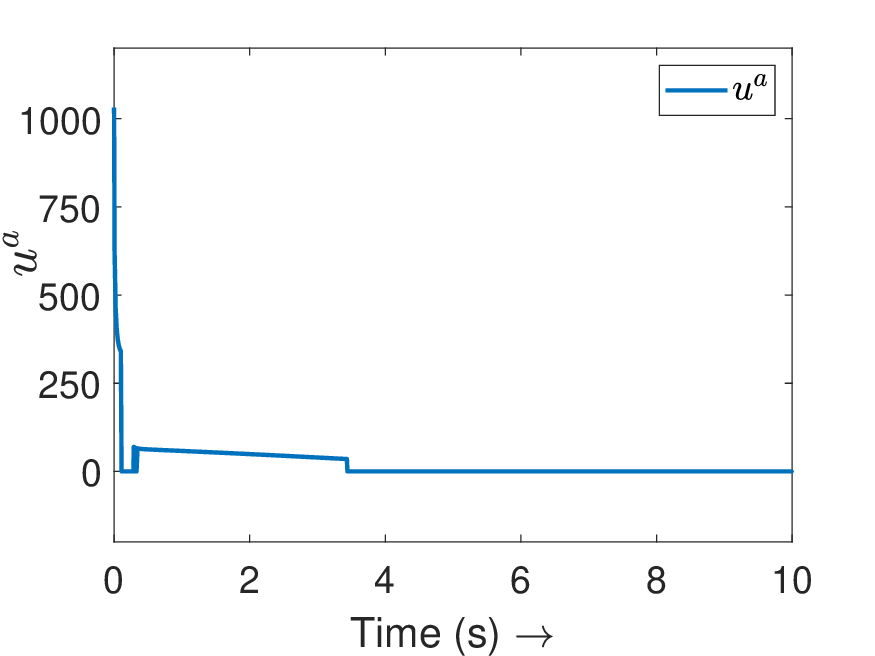}
		\caption{$u^a_2$}
		\label{fig:Passive_attack}
	\end{subfigure}
	\caption{Agent 2 under undetectable attack for $t \in [0,10]$.}
	\label{fig:Attackud}
\end{figure}

As discussed in Definition~\ref{def_detectable}, only the attacks that violate the passivity condition \eqref{passivity_inequality} are detected and mitigated in this paper. Since undetectable attacks do not cause the system to lose passivity (Definition~\ref{def_detectable}), they are not of concern as the system's consensus properties are unaffected by such attacks. We have simulated a result in this direction as shown in Fig.~\ref{fig:Attackud} for the agents in Subsection~\ref{passive_real_poles}, which supports our claim. This happens because the total energy of passive agents approaches zero towards consensus for undetectable attacks, forcing the attack signal to eventually converge to zero to remain undetectable (see Fig.~\ref{fig:Passive_attack}). This shows that the undetectable attacks are decrescent with time and asymptotically converge to zero. 

It would be interesting to also consider the attack on the global communication network and generalize the problem for more complex systems. 


\bibliographystyle{plain}
\bibliography{References}

\begin{thebibliography}{10}

\bibitem{boem2017distributed}
Francesca Boem, Alexander~J Gallo, Giancarlo Ferrari-Trecate, and Thomas
  Parisini.
\newblock A distributed attack detection method for multi-agent systems
  governed by consensus-based control.
\newblock In {\em 2017 IEEE 56th Annual Conference on Decision and Control
  (CDC)}, pages 5961--5966. IEEE, 2017.

\bibitem{carullo2001ultrasonic}
Alessio Carullo, Marco Parvis, et~al.
\newblock An ultrasonic sensor for distance measurement in automotive
  applications.
\newblock {\em IEEE Sensors journal}, 1(2):143, 2001.

\bibitem{chen1999linear}
Chi-Tsong Chen.
\newblock {\em Linear System Theory and Design}.
\newblock Oxford University Press, 1999.

\bibitem{chopra2012output}
Nikhil Chopra.
\newblock Output synchronization on strongly connected graphs.
\newblock {\em IEEE Transactions on Automatic Control}, 57(11):2896--2901,
  2012.

\bibitem{eyisi2014energy}
Emeka Eyisi and Xenofon Koutsoukos.
\newblock Energy-based attack detection in networked control systems.
\newblock In {\em Proceedings of the 3rd international conference on High
  confidence networked systems}, pages 115--124, 2014.

\bibitem{forouzan2008cryptography}
Behrouz~A. Forouzan.
\newblock {\em Cryptography and Network Security}.
\newblock McGraw-Hill Forouzan networking series. McGraw-Hill, 2008.

\bibitem{godsil2001algebraic}
Chris Godsil and Gordon~F Royle.
\newblock {\em Algebraic graph theory}, volume 207.
\newblock Springer Science \& Business Media, 2001.

\bibitem{guo2022stealthy}
Haibin Guo, Jian Sun, and Zhong-Hua Pang.
\newblock Stealthy false data injection attacks with resource constraints
  against multi-sensor estimation systems.
\newblock {\em ISA transactions}, 2022.

\bibitem{huo2022observer}
Shicheng Huo, Ya~Zhang, Frank~L Lewis, and Changyin Sun.
\newblock Observer-based resilient consensus control for heterogeneous
  multi-agent systems against cyber-attacks.
\newblock {\em IEEE Transactions on Control of Network Systems}, 2022.

\bibitem{ihle2007passivity}
Ivar-Andr{\'e}~F Ihle, Murat Arcak, and Thor~I Fossen.
\newblock Passivity-based designs for synchronized path-following.
\newblock {\em Automatica}, 43(9):1508--1518, 2007.

\bibitem{jin2017adaptive}
Xu~Jin, Wassim~M Haddad, and Tansel Yucelen.
\newblock An adaptive control architecture for mitigating sensor and actuator
  attacks in cyber-physical systems.
\newblock {\em IEEE Transactions on Automatic Control}, 62(11):6058--6064,
  2017.

\bibitem{joo2020resilient}
Youngjun Joo, Zhihua Qu, and Toru Namerikawa.
\newblock Resilient control of cyber-physical system using nonlinear encoding
  signal against system integrity attacks.
\newblock {\em IEEE Transactions on Automatic Control}, 66(9):4334--4341, 2020.

\bibitem{khepera}
K-Team.
\newblock {KHEPERA IV NEW - K-Team Corporation}.
\newblock Accessed on 2023-01-19.

\bibitem{khalil2002nonlinear}
H.K. Khalil.
\newblock {\em Nonlinear Systems}.
\newblock Pearson Education. Prentice Hall, 2002.

\bibitem{khazraei2022attack}
Amir Khazraei and Miroslav Pajic.
\newblock Attack-resilient state estimation with intermittent data
  authentication.
\newblock {\em Automatica}, 138:110035, 2022.

\bibitem{li2011consensus}
Zhongkui Li, Xiangdong Liu, Peng Lin, and Wei Ren.
\newblock Consensus of linear multi-agent systems with reduced-order
  observer-based protocols.
\newblock {\em Systems \& Control Letters}, 60(7):510--516, 2011.

\bibitem{lynch2017modern}
Kevin~M Lynch and Frank~C Park.
\newblock {\em Modern robotics}.
\newblock Cambridge University Press, 2017.

\bibitem{meng2020adaptive}
Min Meng, Gaoxi Xiao, and Beibei Li.
\newblock Adaptive consensus for heterogeneous multi-agent systems under sensor
  and actuator attacks.
\newblock {\em Automatica}, 122:109242, 2020.

\bibitem{nowzari2019event}
Cameron Nowzari, Eloy Garcia, and Jorge Cort{\'e}s.
\newblock Event-triggered communication and control of networked systems for
  multi-agent consensus.
\newblock {\em Automatica}, 105:1--27, 2019.

\bibitem{peng2020switching}
Chen Peng and Hongtao Sun.
\newblock Switching-like event-triggered control for networked control systems
  under malicious denial of service attacks.
\newblock {\em IEEE Transactions on Automatic Control}, 65(9):3943--3949, 2020.

\bibitem{petri2021event}
Elena Petri, Romain Postoyan, Daniele Astolfi, D~Ne{\v{s}}i{\'c}, and
  WP~Maurice~H Heemels.
\newblock Event-triggered observer design for linear systems.
\newblock In {\em 2021 60th IEEE Conference on Decision and Control (CDC)},
  pages 546--551. IEEE, 2021.

\bibitem{qin2009observer}
Li~Qin, Zhang Qingling, Zhang Yanjuan, and An~Yichun.
\newblock Observer-based passive control for descriptor systems with
  time-delay.
\newblock {\em Journal of Systems Engineering and Electronics}, 20(1):120--128,
  2009.

\bibitem{rao1971generalized}
C.R. Rao, Sujit~Kumar Mitra, and J.K. Mitra.
\newblock {\em Generalized Inverse of Matrices and Its Applications}.
\newblock Probability and Statistics Series. Wiley, 1971.

\bibitem{shi2014event}
Dawei Shi, Tongwen Chen, and Ling Shi.
\newblock Event-triggered maximum likelihood state estimation.
\newblock {\em Automatica}, 50(1):247--254, 2014.

\bibitem{tan2020brief}
Sen Tan, Josep~M Guerrero, Peilin Xie, Renke Han, and Juan~C Vasquez.
\newblock Brief survey on attack detection methods for cyber-physical systems.
\newblock {\em IEEE Systems Journal}, 14(4):5329--5339, 2020.

\bibitem{trentelman1997storage}
Harry~L Trentelman and Jan~C Willems.
\newblock Storage functions for dissipative linear systems are quadratic state
  functions.
\newblock In {\em Proceedings of the 36th IEEE Conference on Decision and
  Control}, volume~1, pages 42--47. IEEE, 1997.

\bibitem{xia2016synchronization}
Tian Xia and Luca Scardovi.
\newblock Synchronization of linear time-invariant systems on rooted graphs.
\newblock In {\em 2016 IEEE 55th Conference on Decision and Control (CDC)},
  pages 4376--4381. IEEE, 2016.

\bibitem{yan2017resilient}
Yang Yan, Panos Antsaklis, and Vijay Gupta.
\newblock A resilient design for cyber physical systems under attack.
\newblock In {\em 2017 American Control Conference (ACC)}, pages 4418--4423.
  IEEE, 2017.

\end{thebibliography}

\end{document}